\documentclass[journal,letterpaper,draft,onecolumn]{IEEEtran}

\usepackage{url}

\usepackage{cite}

\usepackage{amsmath}
\usepackage{amsfonts}
\usepackage{amssymb}
\usepackage{bm}

\usepackage{amsthm}
\theoremstyle{plain}
\newtheorem{theorem}{Theorem}[section]
\newtheorem{lem}[theorem]{Lemma}

\theoremstyle{definition}
\newtheorem{defn}{Definition}[section]

\theoremstyle{remark}
\newtheorem*{rem}{Remark}

\usepackage{graphicx}

\title{An Interpretation of the Moore-Penrose Generalized Inverse of a Singular Fisher Information Matrix}

\author{Yen-Huan~Li,~\IEEEmembership{Member,~IEEE,}
        Ping-Cheng~Yeh,~\IEEEmembership{Member,~IEEE,}
\thanks{\copyright 2012 IEEE. Personal use of this material is permitted. Permission from IEEE must be obtained for all other uses, in any current or future media, including reprinting/republishing this material for advertising or promotional purposes, creating new collective works, for resale or redistribution to servers or lists, or reuse of any copyrighted component of this work in other works.}
\thanks{Y.-H. Li is with the Research Center for Information Technology Innovation, Academia Sinica, Taipei, Taiwan (e-mail: yenhuan.li@gmail.com).}
\thanks{P.-C. Yeh is with the Department of Electrical Engineering and the Graduate Institute of Communication Engineering, National Taiwan University, Taipei, Taiwan (e-mail: pcyeh@cc.ee.ntu.edu.tw).}}

\date{July 1, 2012}

\begin{document}
\maketitle
\begin{abstract}
It is proved that in a non-Bayesian parametric estimation problem, if the Fisher information matrix (FIM) is singular, unbiased estimators for the unknown parameter will not exist. Cram\'{e}r-Rao bound (CRB), a popular tool to lower bound the variances of unbiased estimators, seems inapplicable in such situations. In this paper, we show that the Moore-Penrose generalized inverse of a singular FIM can be interpreted as the CRB corresponding to the minimum variance among all choices of minimum constraint functions. This result ensures the logical validity of applying the Moore-Penrose generalized inverse of an FIM as the covariance lower bound when the FIM is singular. Furthermore, the result can be applied as a performance bound on the joint design of constraint functions and unbiased estimators.
\end{abstract}

\section{Introduction} \label{section_intro}
An interpretation of the Moore-Penrose generalized inverse \cite{Horn1985} of a singular Fisher information matrix (FIM) is presented in this paper, from the perspective of Cram\'{e}r-Rao bound (CRB). CRB is a lower bound on the covariance matrix of any unbiased estimator in a non-Bayesian parametric estimation problem \cite{Cramer1999,Rao1945}, and is a popular tool to evaluate the optimal mean-square error (MSE) performance of estimators in various applications \cite{Kay1993,VanTrees2001}. The most general form of CRB says that the covariance matrix of any unbiased estimator is lower bounded by the generalized inverse of the Fisher information matrix \cite{Rao1973}. This general form of CRB holds for both singular and non-singular FIMs.

There are, however, facts in literature that render the application of CRB questionable when the FIM is singular. Rothenberg proves in \cite{Rothenberg1971} that under some regularity conditions, the non-singularity of the FIM is equivalent to the local identifiability of the parameter to be estimated\footnote{A parameter $\bm{\theta}$ is locally identifiable if there exists an open neighbourhood $\bm{\Theta}$ of $\bm{\theta}$ such that no other $\bm{\theta}' \in \bm{\Theta}$ is observationally equivalent to $\bm{\theta}$.}; Stoica \textit{et al.} prove in \cite{Stoica1982,Stoica1998} that unbiased estimators with finite variances do not exist when the FIM is singular, except for some ``unusual'' conditions\footnote{The ``unusual'' conditions suggest that if the FIM is singular, only unbiased estimators for some functions of the unknown parameter, instead of the unknown parameter itself, may exist with finite variances.}. If the parameter to be estimated is locally non-identifiable, or all of the unbiased estimators will have infinite variances, it seems meaningless to discuss the performances of unbiased estimators.

As mentioned in \cite{Stoica2001}, one may change the nature of an estimation problem to allow the existence of an estimator with finite variance. There are three approaches. The first approach is to introduce \textit{a priori} information about the probability distribution of the parameter to be estimated; in this way the estimation problem becomes a Bayesian one. There are abundant literature on Bayesian statistics \cite{Berger1985} and performance bounds \cite{VanTrees2007}. \textit{A priori} information about the probability distribution of the unknown parameter, however, is not always already known. The second approach is to consider \emph{biased} estimators instead of \emph{unbiased} estimators. In \cite{Stoica2001}, the necessary condition for the bias function to ensure the existence of an unbiased estimator with finite variance is derived. The authors of a recent paper derive the bias function that leads to the minimum trace of the resulting CRB, a lower bound on the total variances of estimators \cite{Song2009}. There are a number of situations, however, where biased estimators are not preferred. For example, almost all estimation problems encountered in the design of a communication system, including the estimation of carrier phases and symbol timing for synchronization, the estimation of channel responses for equalization, etc., require unbiased estimators. The third approach is to put or to exploit some deterministic constraints on the parameter to be estimated. The deterministic constraints result in a parametric estimation problem with reduced dimension, where an unbiased estimator with finite variance may exist. We focus on the third approach in this paper.

Take blind channel estimation problems for example \cite{Tong1998}. The goal of blind channel estimation is to estimate the channel response $\bm{h}$ from $\bm{y} = \bm{s} * \bm{h} + \bm{n}$, the convolution of the channel response $\bm{h}$ and the input data sequence $\bm{s}$, corrupted by an additive noise $\bm{n}$. The unknown parameter $\bm{\theta} \triangleq ( \bm{s}, \bm{h} )$ is not identifiable since $( \alpha \bm{s}, \frac{1}{\alpha} \bm{h} )$ and $( \bm{s}, \bm{h} )$ are observationally equivalent for any constant $\alpha \neq 0$, so unbiased estimators do not exist. Practically this so-called scalar ambiguity problem is resolved by assigning a pre-determined value to one of the element of $\bm{s}$ \cite{Su2007}. That is, a constraint function $f( \bm{\theta} ) \triangleq s_n - c = 0$ is put on the parameter $\bm{\theta}$, where $s_n$ denotes the $n$th element of $\bm{s}$ and $c$ is some pre-determined constant. This is exactly the third approach mentioned above.

CRB for constrained parameters is already derived in \cite{Stoica1998,Gorman1990,Marzetta1993}. The value of the constrained CRB depends on the choice of the constraint function; different constraint functions lead to different values of the CRB. This bound is useful when the constraint function is exogenously given, but there are situations where we are able to modify the constraint function. Take blind channel estimation problems for example again. Suppose an engineer chooses the constraint function as $f( \bm{\theta} ) \triangleq s_1 - c = 0$ and designs an unbiased estimator corresponding to this constraint function, and finds the resulting MSE, although almost achieving the constrained CRB with respect to the constraint function, is still unsatisfactory compared with the target value. How can the engineer tell the unsatisfactory result is caused by the inappropriate choice of the constraint function, or simply because the target value is not attainable for \emph{any} choice of the constraint function?

The main contribution of this paper is the following theorem. The Moore-Penrose generalized inverse of a singular FIM is \emph{the constrained CRB corresponding to the minimum variance among all choices of minimum constraint functions}. According to the theorem, the logical validity of using the Moore-Penrose generalized inverse of a singular FIM as a covariance lower bound for unbiased estimators is justified, and a CRB for the joint design of the unbiased estimator and the constraint function is obtained. In addition to a performance bound, we also provide a sufficient condition for a constraint function to achieve the bound, which is an affine function of the parameter to be estimated. The above results facilitate future researches on the optimal joint design of constraint functions and unbiased estimators.

A mathematical definition of minimum constraint functions will be given in Section \ref{section_min_constraint}, but the meaning is conceptually easy to understand. In blind channel estimation problems, only a one-dimensional constraint on $\bm{\theta}$ is needed to resolve the scalar ambiguity, such as $f( \bm{\theta} ) = s_n - 1$, and any constraint function $\bm{f}$ that is essentially a one-dimensional constraint is a minimum constraint function as long as the constrained CRB exists.

The rest of the paper is organized as follows. The necessary background knowledge is given in Section \ref{section_preliminary}. Then we show that the Moore-Penrose generalized inverse of the FIM can be viewed as a CRB for constrained parameters with some constraint function in Section \ref{section_J_dagger_is_constrained_crb}. Section \ref{section_interpret_J_dagger} is divided into two sub-sections. In the first sub-section we give the definition of minimum constraint functions and justify its meaning. In the second sub-section we prove the main result of this paper, that the Moore-Penrose generalized inverse of the FIM is the CRB corresponding to the minimum variance among all choices of minimum constraint functions. Conclusions and some discussions are presented in Section \ref{section_conclusions}.

\subsection*{Notation}
Bold-faced lower case letters represent column vectors, and bold-faced upper case letters are matrices. Superscripts $\bm{v}^T$, $\bm{M}^{-1}$, and $\bm{M}^{\dagger}$ denote the transpose, inverse, and the Moore-Penrose generalized inverse of the corresponding vector or matrix. The vector $\mathsf{E}\left[ \bm{v} \right]$ denotes the expectation of the random vector $\bm{v}$, and $\mathsf{E}\left[ \bm{M} \right]$ denotes the expectation of the random matrix $\bm{M}$. The matrix $\mathsf{cov}(\bm{u}, \bm{v})$ is defined as $\mathsf{cov}(\bm{u}, \bm{v}) \triangleq \mathsf{E} \left[ ( \bm{u} - \mathsf{E}( \bm{u} ) ) ( \bm{v} - \mathsf{E}( \bm{v} ) )^T \right]$, which is the cross-covariance matrix of random vectors $\bm{u}$ and $\bm{v}$. We use the notation $\bm{A} \geq \bm{B}$ to mean that $\bm{A} - \bm{B}$ is a nonnegative-definite matrix. The notation $\mathsf{rank}\bm{M}$ denotes the rank of the matrix $\bm{M}$.

\section{Preliminaries} \label{section_preliminary}
In this section, some background knowledge required to begin the discussions in the following sections is presented. We restrict our attention to the case of \emph{unbiased estimators for the unknown parameter}, so the theorems presented in this section may be simplified versions of those on the original papers.

When we refer to the \emph{CRB for unconstrained parameters}, we mean the following theorem.

\begin{theorem}[CRB for unconstrained parameters] \label{them_crlb}
Let $\hat{\bm{\theta}}$ be an unbiased estimator of a real unknown parameter $\bm{\theta}$ based on real observation $\bm{y}$, which is characterized by its probability density function (pdf) $p( \bm{y}; \bm{\theta} )$. Then for any such $\hat{\bm{\theta}}$, we have
\begin{equation}
\mathsf{cov}\left( \hat{\bm{\theta}}, \hat{\bm{\theta}} \right) \geq \bm{J}^{\dagger}, \notag
\end{equation}
where $\bm{J}$ is the FIM defined as
\begin{equation} \label{eq_FIM}
\bm{J} \triangleq \mathsf{E}\left[ \frac{\partial \ln p}{\partial \bm{\theta}} \frac{\partial \ln p}{\partial \bm{\theta}^T} \right]. 
\end{equation}
The equality is achieved if and only if 
\begin{equation}
\hat{\bm{\theta}} - \bm{\theta} = \bm{J}^{\dagger} \frac{\partial \log p}{\partial \bm{\theta}} \notag
\end{equation}
in the mean square sense.
\end{theorem}
\begin{proof}
See \cite{Rao1973,Kay1993,VanTrees2001}.
\end{proof}

The above theorem is always correct given that unbiased estimators exist. Stoica \textit{et al.}, however, prove the following theorem in \cite{Stoica2001}.

\begin{theorem}
If the information matrix $\bm{J}$ is singular, then there does not exist an unbiased estimator with finite variance.
\end{theorem}
\begin{proof}
See \cite{Stoica2001}\footnote{When we restrict our attention to unbiased estimators for the unknown parameter only, the condition for the existence of an unbiased estimator with finite variance in \cite{Stoica2001} becomes $\bm{J} \bm{J}^{\dagger} = \bm{I}$, which is impossible for singular FIMs.}.
\end{proof}

That is, there does not exist any finite unbiased estimator $\hat{\bm{\theta}}$ if the FIM is singular, so the CRB fails to provide any useful information.

When we refer to the \emph{CRB for constrained parameters}, we mean the following theorem.

\begin{theorem}[CRB for constrained parameters] \label{thm_constrained_crlb}
Let $\hat{\bm{\theta}}$ be an unbiased estimator of an unknown parameter $\bm{\theta} \in \mathbb{R}^n$ based on real observation $\bm{y}$, which is characterized by its pdf $p( \bm{y}; \bm{\theta} )$. Furthermore, we require the parameter $\bm{\theta}$ to satisfy a possibly non-affine constraint function $\bm{f}: \mathbb{R}^n \to \mathbb{R}^m$, $m\leq n$, 
\begin{equation}
\bm{f}( \bm{\theta} ) = \bm{0}. \notag
\end{equation}
Assume that $\partial \bm{f} / \partial \bm{\theta}^T$ is full rank. Choose a matrix $\bm{U}$ with $(n-m)$ orthonormal columns such that
\begin{equation}
\frac{\partial \bm{f}}{ \partial \bm{\theta}^T } \bm{U} = \bm{0}. \notag
\end{equation}
If $\bm{U}^T \bm{J} \bm{U}$ is nonsingular, then
\begin{equation}
\mathsf{cov}\left( \hat{\bm{\theta}}, \hat{\bm{\theta}} \right) \geq \bm{U} \left( \bm{U}^T \bm{J} \bm{U} \right)^{-1} \bm{U}^T, \notag
\end{equation}
where $\bm{J}$ is the FIM defined as in (\ref{eq_FIM}). The equality is achieved if and only if 
\begin{equation}
\hat{\bm{\theta}} - \bm{\theta} = \bm{U} ( \bm{U}^T \bm{J} \bm{U} )^{-1} \bm{U}^T \frac{\partial \log p}{\partial \bm{\theta}} \notag
\end{equation}
in the mean square sense.
\end{theorem}
\begin{proof}
See \cite{Stoica1998}.
\end{proof}

The following theorem gives a necessary and sufficient condition for the existence of a finite constrained CRB.

\begin{theorem} \label{thm_no_estimator_constrained_crb}
The constrained CRB is finite if and only if the matrix $\bm{U}^T \bm{J} \bm{U}$ is non-singular.
\end{theorem}
\begin{proof}
See \cite{Stoica1998}.
\end{proof}

Now we are able to discuss the relationship between the Moore-Penrose generalized inverse of an FIM and constrained CRB.

\section{$\bm{J}^{\dagger}$ as a CRB for Constrained Parameters} \label{section_J_dagger_is_constrained_crb}

The main result of this section is the following theorem. 
\begin{theorem} \label{thm_pseudoinverse_constrained_crlb}
Let the FIM $\bm{J}$ be singular with rank $r$, and let the singular value decomposition (SVD) of $\bm{J}$ be
\begin{equation}
\bm{J} = \left[ \begin{array}{cc} \bm{U}_r & \overline{\bm{U}}_r \end{array} \right] \left[ \begin{array}{cc} \bm{\Sigma} & \bm{0} \\ \bm{0} & \bm{0}  \end{array} \right] \left[ \begin{array}{c} \bm{U}_r^T \\ \overline{\bm{U}}_r^T \end{array} \right], \label{eq_svd_J}
\end{equation}
the diagonal elements of $\bm{\Sigma}$ being nonzero. Then $\bm{J}^{\dagger}$ is a CRB for constrained parameters with constraint function
\begin{equation}
\bm{f}( \boldsymbol{\theta} ) = \overline{\boldsymbol{U}}_r^T \boldsymbol{\theta} + \boldsymbol{C} = \boldsymbol{0} \label{eq_exp_min_constraint_func}
\end{equation}
for some constant matrix $\bm{C}$.
\end{theorem}

To prove the theorem, we first prove the following lemma.

\begin{lem} \label{lem_pseudo_inverse_by_Us}
Let the SVD of a Hermitian matrix $\bm{J}$ be the same as in (\ref{eq_svd_J}). Then 
\begin{equation}
\bm{J}^{\dagger} = \bm{U}_r \left( \bm{U}_r^T \bm{J} \bm{U}_r \right)^{-1} \bm{U_r}^T. \label{eq_pseudoinverse_us}
\end{equation}
\end{lem}
\begin{proof}
Substitute $\bm{J}$ as $\bm{J} = \bm{U}_r \bm{\Sigma} \bm{U}_r^T$ into (\ref{eq_pseudoinverse_us}).
\end{proof}

Now we are able to prove Theorem \ref{thm_pseudoinverse_constrained_crlb}.

\begin{proof}[Proof for Theorem \ref{thm_pseudoinverse_constrained_crlb}]
By examining the lemma and Theorem \ref{thm_constrained_crlb}, we can think of $\bm{J}^{\dagger}$ as a constrained CRB with some constraint function $\bm{f}(\bm{\theta})$ such that
\begin{equation}
\frac{\partial \bm{f}}{ \partial \bm{\theta}^T } \bm{U}_r = \bm{0}. \label{eq_constraint_f}
\end{equation}

Since $\overline{\bm{U}}_r^T \bm{U}_r = \mathbf{0}$ by the definition of SVD, 
a constraint function $\bm{f}$ that satisfies (\ref{eq_constraint_f}) can be chosen such that
\begin{equation}
\frac{\partial \bm{f}}{ \partial \bm{\theta}^T } = \overline{\bm{U}}_r^T. \notag
\end{equation}
The above equation can be satisfied by an affine constraint function, 
\begin{equation}
\bm{f}( \boldsymbol{\theta} ) = \overline{\bm{U}}_r^T \boldsymbol{\theta} + \boldsymbol{C} = \boldsymbol{0}, \notag
\end{equation}
and the theorem is proved.
\end{proof}

\begin{rem}
In fact, any constraint function satisfying
\begin{equation}
\frac{\partial \bm{f}}{\partial \bm{\theta}^T} \bm{U}_r = \bm{0} \notag
\end{equation}
leads to the same constrained CRB, $\bm{J}^\dagger$.
\end{rem}

\section{Interpretation of $\bm{J}^{\dagger}$ as a CRB for Constrained Parameters} \label{section_interpret_J_dagger}
In this section we prove that $\bm{J}^{\dagger}$ is not only a CRB for constrained parameters, but \emph{the CRB corresponding to the minimum variance among all choices of minimum constraint functions}. We first give a definition of minimum constraint functions, and then prove the claim. 

\subsection{Definition of Minimum Constraint Functions} \label{section_min_constraint}
Minimum constraint functions are defined as follows.

\begin{defn} \label{def_min_constraint}
A differentiable constraint function $\bm{f}: \mathbb{R}^n \to \mathbb{R}^m$, $m\leq n$, for a non-Bayesian parametric estimation problem with a singular FIM $\bm{J}$ is a minimum constraint if 
\begin{enumerate}
\item $\partial \bm{f} / \partial \bm{\theta}^T$ is full rank, 
\item $\bm{U}^T \bm{J} \bm{U}$ is nonsingular, and 
\item $\mathsf{rank}\ \partial \bm{f} / \partial \bm{\theta}^T + \mathsf{rank} \ \bm{J} = n$,
\end{enumerate}
where $\bm{U}$ is chosen as in Theorem \ref{thm_constrained_crlb}.
\end{defn}

The first requirement is to ensure that $\bm{f}$ does not contain any redundant constraints \cite{Gorman1990, Marzetta1993}. The second requirement is to ensure the existence of a finite CRB according to Theorem \ref{thm_no_estimator_constrained_crb}. The third requirement means that $\bm{f}$ contains the minimum number of independent constraints. Take blind channel estimation problems as an example. From discussions in Section \ref{section_intro} we know that once we choose one symbol as a pilot symbol with some pre-determined value, we eliminate the scalar ambiguity and thus an unbiased estimator exists. Note that the nullity of the FIM is also one \cite{Carvalho1997, Barbarossa2002}. We can see the third requirement holds.

Now we give a formal proof that if the first two requirements are satisfied, then the third requirement ensures that $\bm{f}$ contains the minimum number of independent constraints. 

\begin{theorem}
For any constraint function $\bm{f}$ in Definition \ref{def_min_constraint} that satisfies the first and the second requirements, 
\begin{equation}
\min_{\bm{f}} \mathsf{rank}\ \frac{\partial \bm{f}}{\partial \bm{\theta}^T} = n - \mathsf{rank}\ \bm{J}. \notag
\end{equation}
\end{theorem}
\begin{proof}
First we show that in order to satisfy the first and the second requirements, 
\begin{equation}
\mathsf{rank}\ \frac{\partial \bm{f}}{\partial \bm{\theta}^T} \geq n - \mathsf{rank}\ \bm{J}, \label{eq_min_rank_F}
\end{equation}
and then we show that the equality is achievable.

If 
\begin{equation}
\mathsf{rank}\ \frac{\partial \bm{f}}{\partial \bm{\theta}^T} < n - \mathsf{rank}\ \bm{J}, \notag
\end{equation}
by the definition of $\bm{U}$ (see Theorem \ref{thm_constrained_crlb}), $\bm{U}$ is a $n$-by-$(\mathsf{rank}\ \bm{U})$ matrix with
\begin{equation} 
n \geq \mathsf{rank}\ \bm{U} > \mathsf{rank}\ \bm{J}. \label{eq_n_rankU_rankJ}
\end{equation}

By the fact that
\begin{align}
\mathsf{rank}\ \bm{\bm{U}^T \bm{J} \bm{U}} &\leq \min \{ \mathsf{rank}\ \bm{U}, \mathsf{rank}\ \bm{J} \} \notag \\
                                           &\leq \mathsf{rank}\ \bm{J} < \mathsf{rank}\ \bm{U},  \notag
\end{align}
where the last inequality follows by (\ref{eq_n_rankU_rankJ}), and noting that $\bm{U}^T \bm{J} \bm{U}$ is a $(\mathsf{rank}\ \bm{U})$-by-$(\mathsf{rank}\ \bm{U})$ square matrix, $\bm{U}^T \bm{J} \bm{U}$ cannot be full-rank. Thus (\ref{eq_min_rank_F}) is proved.

The achievability of equality in (\ref{eq_min_rank_F}) is easy to prove. Choose the constraint function $\bm{f}$ as in (\ref{eq_exp_min_constraint_func}), and we can see such a constraint function satisfies all of the requirements of a minimum constraint function.
\end{proof}

By the above theorem we can see the third requirement is in fact requiring $\partial \bm{f} / \partial \bm{\theta}^T$ to have the minimum rank. The reason why such a constraint function $\bm{f}$ can be considered as the constraint function with \emph{minimum constraints} can be found by the following theorem.

\begin{theorem}
Let $A\subset \mathbb{R}^n$ be open and let $\bm{f}: A\to\mathbb{R}^m$, $m \leq n$, be a differentiable function such that $\partial \bm{f} / \partial \bm{\theta}^T$ has rank $m$ whenever $\bm{f}( \bm{x} ) = \bm{0}$. Then $\bm{f}( \bm{x} ) = \bm{0}$ implicitly defines an $(n-m)$-dimensional manifold in $\mathbb{R}^n$.
\end{theorem}
\begin{proof}
See \cite{Spivak1965}.
\end{proof}

Constraint functions $\bm{f}$ with the minimum $\mathsf{rank}\ \partial \bm{f} / \partial \bm{\theta}^T$ ensures that the resulting manifolds have the maximal degree of freedom, so we call them minimum constraint functions.

\subsection{$\bm{J}^{\dagger}$ is the CRB corresponding to the minimum variance among all choices of minimum constraint functions.}
In this sub-section, we prove the claim that $\bm{J}^{\dagger}$ is the CRB corresponding to the minimum variance among all choices of minimum constraint functions. For convenience, the $i$th largest eigenvalue of a matrix $\bm{M}$ is denoted by $\lambda_i( \bm{M} )$ in the following discussions.

The main result of this subsection is the following theorem.

\begin{theorem} \label{thm_min_crb}
In Theorem \ref{thm_constrained_crlb}, if $\bm{f}$ is a minimum constraint function, then 
\begin{equation}
\mathsf{tr}\left( \mathsf{cov}\left[ \hat{\bm{\theta}}, \hat{\bm{\theta}} \right] \right) \geq \mathsf{tr} \left( \bm{J}^\dagger \right). \notag
\end{equation}
Furthermore, equality can be achieved by choosing the constraint function $\bm{f}$ as in Theorem \ref{thm_pseudoinverse_constrained_crlb}.
\end{theorem}

Note that the trace of a covariance matrix is the sum of the variances of the elements of $\hat{ \bm{\theta} }$. In this way, we have proved that the Moore-Penrose generalized inverse of the FIM is the CRB corresponding to the minimum variance among all choices of minimum constraint functions.

Theorem \ref{thm_min_crb} is in fact a corollary of the following theorem.

\begin{theorem} \label{thm_min_is_J_dagger}
Let the SVD of an $n$-by-$n$ nonnegative definite matrix $\bm{J}$ with rank $r$ be 
\begin{equation}
\bm{J} = \left[ \begin{array}{cc} \bm{U}_r & \overline{\bm{U}}_r \end{array} \right] \left[ \begin{array}{cc} \bm{\Sigma} & \bm{0} \\ \bm{0} & \bm{0}  \end{array} \right] \left[ \begin{array}{c} \bm{U}_r^T \\ \overline{\bm{U}}_r^T \end{array} \right],  \notag
\end{equation}
where $\bm{\Sigma}$ is a diagonal matrix with nonnegative diagonal elements. Then
\begin{align}
\lambda_i( \bm{V} \left( \bm{V}^T \bm{J} \bm{V} \right)^{-1} \bm{V}^T ) 
&\geq \lambda_i( \bm{U}_r \left( \bm{U}_r^T \bm{J} \bm{U}_r \right)^{-1} \bm{U}_r^T ) \notag \\
&= \lambda_i( \bm{J}^\dagger ) \quad \forall i \notag
\end{align}
for any matrix $\bm{V}$ with the same size as $\bm{U}_r$ such that $\bm{V}^T \bm{V} = \bm{I}$.
\end{theorem}
\begin{proof}
See Appendix.
\end{proof}

If the above theorem holds, then Theorem \ref{thm_min_crb} can be proved as follows.

\begin{proof}[Proof for Theorem \ref{thm_min_crb}]
Note that the FIM $\bm{J}$ is nonnegative definite, and the resulting $\bm{U}$ (see Theorem \ref{thm_constrained_crlb}) for every minimum constraint function $\bm{f}$ should have the same size as $\bm{U_r}$ in Theorem \ref{thm_min_is_J_dagger}, so the above theorem applies. Noting that $\bm{U}_r \left( \bm{U}_r^T \bm{J} \bm{U}_r \right)^{-1} \bm{U}_r^T = \bm{J}^{\dagger}$ according to Lemma \ref{lem_pseudo_inverse_by_Us}, the theorem follows because trace equals to the sum of eigenvalues.
\end{proof}

\begin{rem}
One may expect that the inequality
\begin{equation}
\bm{V}( \bm{V}^T \bm{J} \bm{V} )^{-1} \bm{V}^T \geq \bm{J}^\dagger \notag
\end{equation}
holds for matrices $\bm{V}$ and $\bm{J}$ defined as in Theorem \ref{thm_min_is_J_dagger}, but in general this matrix inequality does not hold. A counterexample is when 
\begin{equation}
\bm{J} = \left[ \begin{array}{cccc}
1 & 0 & 0 & 0 \\
0 & 1 & 0 & 0 \\
0 & 0 & 0 & 0 \\
0 & 0 & 0 & 0
\end{array} \right], 
\bm{V} = \frac{1}{2} \left[ \begin{array}{rr}
-1 & 1 \\
-1 & -1 \\
-1 & 1 \\
-1 & -1
\end{array} \right]. \notag
\end{equation}
\end{rem}

\section{Conclusions and Discussions} \label{section_conclusions}
We have proved the main theorem in this paper: The Moore-Penrose generalized inverse of a singular FIM is the CRB corresponding to the minimum variance among all choices of minimum constraint functions. According to the theorem, the logical validity of using the Moore-Penrose generalized inverse of a singular FIM as a CRB is justified, and a CRB for the joint design of the unbiased estimator and the constraint function is obtained. In addition to a performance bound, we also derive the sufficient condition for a constraint function to achieve the bound, which is an affine function of the parameter to be estimated. The above results facilitate future researches on the optimal joint design of constraint functions and unbiased estimators.

One possible extension of this study is to generalize the concept of a minimum constraint function to higher dimensional cases. To be more precise, it may be possible to consider the minimum CRB when the rank of the constraint function $\bm{f}$ is larger than $( n - \mathsf{rank}\bm{J} )$ (cf. Definition \ref{def_min_constraint}). This extension may be of practical interest because the CRB, if derived, could be useful in the study of semi-blind channel estimation problems, where more than one pilot symbols exist \cite{Tong2004}.

\section*{Acknowledgements}
The authors would like to thank anonymous reviewers for their informative comments.

\appendix \label{app_min_is_J_dagger}
Observing that
\begin{align}
\lambda_i \left(\bm{V} \left( \bm{V}^T \bm{J} \bm{V} \right)^{-1} \bm{V}^T \right) = \lambda_i \left( \bm{U}_r \left( \bm{U}_r^T \bm{J} \bm{U}_r \right)^{-1} \bm{U}_r^T \right) = 0 \notag
\end{align}
for all $i\in \{ r+1, r+2, \ldots, n \}$, and 
\begin{gather}
\lambda_i \left(\bm{V} \left( \bm{V}^T \bm{J} \bm{V} \right)^{-1} \bm{V}^T \right) = \lambda_i \left(\left( \bm{V}^T \bm{J} \bm{V} \right)^{-1} \right), \notag \\
\lambda_i \left( \bm{U}_r \left( \bm{U}_r^T \bm{J} \bm{U}_r \right)^{-1} \bm{U}_r^T \right) = \lambda_i \left( \left( \bm{U}_r^T \bm{J} \bm{U}_r \right)^{-1} \right) \notag,
\end{gather}
for all $i\in\{ 1, 2, \ldots, r \}$, it suffices to prove 
\begin{equation}
\lambda_j \left(\left( \bm{V}^T \bm{J} \bm{V} \right)^{-1} \right) \geq \lambda_j \left( \left( \bm{U}_s^T \bm{J} \bm{U}_s \right)^{-1} \right) \notag
\end{equation}
for all $j \in \{ 1, 2, \ldots, r \}$, or equivalently, 
\begin{equation}
\lambda_k \left( \bm{V}^T \bm{J} \bm{V} \right) \leq \lambda_k \left( \bm{U}_r^T \bm{J} \bm{U}_r \right) \label{eq_result_of_poincare}
\end{equation}
for all $k\in \{ 1, 2, \ldots, r \}$.

Noting that $\lambda_i \left( \bm{U}_r^T \bm{J} \bm{U}_r \right) = \lambda_i \left( \bm{J} \right)$ because they have the same first $r$ eigenvalues, and by the fact that an FIM is always Hermitian, we can see (\ref{eq_result_of_poincare}) is just a result of Poincar\'{e} separation theorem \cite{Horn1985}. Therefore the theorem follows.

\bibliographystyle{IEEEtran}
\bibliography{list}

\end{document}